\documentclass[conference]{IEEEtran}

\usepackage[firstpage]{draftwatermark}
\SetWatermarkText{\hspace*{6.2in}\raisebox{7in}
  {\includegraphics[scale=0.1]{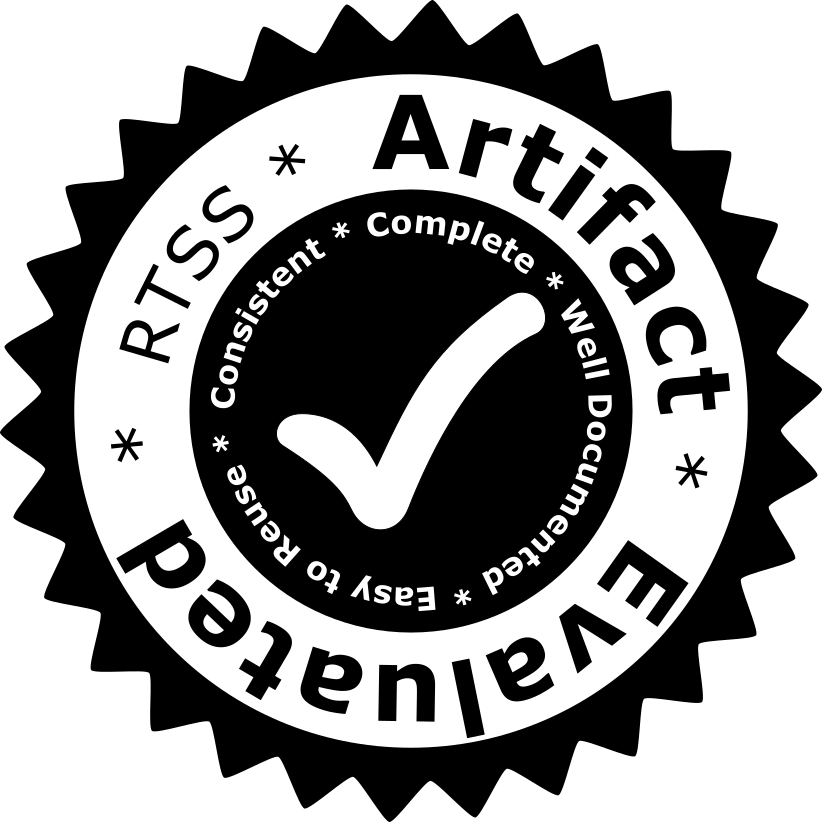}}}
\SetWatermarkAngle{0}

\IEEEoverridecommandlockouts
\usepackage{cite}
\usepackage{amsmath,amssymb,amsfonts}
\usepackage[linesnumbered, ruled]{algorithm2e}
\usepackage{graphicx}
\usepackage{textcomp}
\usepackage{xcolor}
\usepackage{amsthm}
\usepackage{enumerate}
\usepackage{scalerel}
\usepackage{array}
\usepackage[hidelinks]{hyperref}
\usepackage{subfigure}
\usepackage{bbm}
\usepackage{dsfont}
\usepackage{url}
\usepackage{booktabs}

\def\BibTeX{{\rm B\kern-.05em{\sc i\kern-.025em b}\kern-.08em
    T\kern-.1667em\lower.7ex\hbox{E}\kern-.125emX}}

\newtheorem{theorem}{Theorem}

\newtheorem{definition}{Definition}

\begin{document}

\title{Real-Time Service Subscription and Adaptive Offloading Control in Vehicular Edge Computing
\thanks{This work was supported by the MoE Tier-2 grant MOE-T2EP20221-0006 and the MoE Tier-2 grant MOE-T2EP20224-0007.}
}

\author{
\IEEEauthorblockN{Chuanchao Gao, Arvind Easwaran}
\IEEEauthorblockA{\textit{College of Computing and Data Science}}
\IEEEauthorblockA{\textit{Energy Research Institute @ NTU, Interdisciplinary Graduate Programme}}
\IEEEauthorblockA{
\textit{Nanyang Technological University, Singapore}\\
gaoc0008@e.ntu.edu.sg, arvinde@ntu.edu.sg}
}

\maketitle

\begin{abstract}
Vehicular Edge Computing (VEC) has emerged as a promising paradigm for enhancing the computational efficiency and service quality in intelligent transportation systems by enabling vehicles to wirelessly offload computation-intensive tasks to nearby Roadside Units. However, efficient task offloading and resource allocation for time-critical applications in VEC remain challenging due to constrained network bandwidth and computational resources, stringent task deadlines, and rapidly changing network conditions. To address these challenges, we formulate a Deadline-Constrained Task Offloading and Resource Allocation Problem (DOAP), denoted as $\mathbf{P}$, in VEC with both bandwidth and computational resource constraints, aiming to maximize the total vehicle utility. To solve $\mathbf{P}$, we propose $\mathtt{SARound}$, an approximation algorithm based on Linear Program rounding and local-ratio techniques,  that improves the best-known approximation ratio for DOAP from $\frac{1}{6}$ to $\frac{1}{4}$. Additionally, we design an online service subscription and offloading control framework to address the challenges of short task deadlines and rapidly changing wireless network conditions. To validate our approach, we develop a comprehensive VEC simulator, VecSim, using the open-source simulation libraries OMNeT++ and Simu5G. VecSim integrates our designed framework to manage the full life-cycle of real-time vehicular tasks. Experimental results, based on profiled object detection applications and real-world taxi trace data, show that $\mathtt{SARound}$ consistently outperforms state-of-the-art baselines under varying network conditions while maintaining runtime efficiency.
\end{abstract}


\addtolength{\textfloatsep}{-2mm}

\section{Introduction}
Internet of Vehicles (IoV), as a derivative technology of Internet of Things, has revolutionized intelligent transportation systems by enabling seamless interaction between vehicles and infrastructure. Leveraging advanced wireless technologies such as the ultra-reliable low-latency communication (URLLC) of 5G, IoV supports a wide range of time-critical vehicle applications, such as autonomous driving, intelligent traffic management, and online navigation \cite{fan2023joint}. These applications typically require substantial computational resources and impose stringent latency requirements, posing significant challenges for resource-constrained onboard processing units.

Vehicular Edge Computing (VEC) systems have emerged as a promising solution to enhance computational efficiency in IoV by enabling vehicles to offload computation-intensive tasks to nearby Roadside Units (RSUs) \cite{dai2018joint}. Each RSU comprise a 5G module that provides wireless bandwidth for task offloading and a server that offers computational resources for task processing. By processing tasks on nearby RSUs, VEC reduces reliance on centralized cloud infrastructure and lowers communication latency between vehicles and servers, making it well-suited for latency-sensitive vehicular tasks.

Offloading tasks to RSUs offers several advantages, such as enhanced task performance, reduced energy consumption, and increased operational safety. Computation-intensive tasks typically consume substantial energy, and offloading these tasks to RSUs can significantly reduce the energy burden on mobile devices, which is important for energy-constrained devices such as unmanned aerial vehicles (UAVs) or robots, especially those powered by energy harvesting or with limited recharging opportunities \cite{oliveira2024internet}. Besides, in intelligent transportation systems, VEC can support collaborative perception services, enabling vehicles to acquire contextual information beyond their sensing range. This extended perception is vital for applications such as anomaly detection, where enhanced accuracy contributes directly to safer and more reliable decision-making by the vehicle’s control system.

Despite the advantages of VEC, the bandwidth and computational capacities of RSUs are typically constrained, necessitating efficient strategies for \textit{task offloading} (determining the RSU to deploy the service for each task) and \textit{resource allocation} (optimizing the allocation of bandwidth and computational resources).
Jointly optimizing task offloading and resource allocation in VEC enables adaptive resource management based on RSU resource availability and wireless channel quality between vehicles and RSUs. For instance, the VEC with abundant computational resources but limited wireless bandwidth can allocate more computational resources and less bandwidth to each task to accommodate more tasks without violating task deadlines. 

Numerous studies have investigated Deadline-constrained task Offloading Problems (DOP) in edge computing systems, typically assuming fixed resource allocation for each task (e.g., \cite{lin2023a, li2022maximizing, li2023throughput, li2024computation, liao2024an, ma2022mobility, xu2024age, yang2024a, zhao2021energy, liu2021energy}). However, the joint optimization of task offloading and resource allocation, referred to as Deadline-constrained task Offloading and resource Allocation Problem (DOAP) \cite{gao2022deadline, fan2023joint, hu2020dynamic, qi2024joint, dai2018joint, li2019cooperative, vu2021optimal, gao2025local}, remains underexplored. Moreover, existing studies primarily focus on either exponential-time exact (optimal) algorithms \cite{vu2021optimal, fan2023joint} or polynomial-time heuristics without performance guarantees \cite{gao2022deadline, hu2020dynamic, qi2024joint, dai2018joint, li2019cooperative}, necessitating efficient strategies that balance computational efficiency and solution quality. While a $\frac{1}{6}$-approximation algorithm for DOAP has been proposed recently \cite{gao2025local}, ensuring that the achieved objective is at least $\frac{1}{6}$ of the optimal value, we show that this guarantee can be further improved in this paper.


Besides the resource utilization efficiency, VEC presents additional challenges that distinguish it from conventional edge computing systems. First, vehicular tasks are highly latency-sensitive and often have ultra-short deadlines (e.g., of the order of $10$ ms \cite{Garcia2021A}). Given such stringent deadlines, the overhead of scheduling algorithm execution and service initialization can consume a significant portion of the available execution window during runtime, making deadline compliance particularly challenging. Second, wireless channel quality in VEC changes rapidly due to vehicle mobility and environmental factors. The 5G network is commonly used in VEC due to its ultra-reliability and low latency. To ensure reliable communication, 5G networks dynamically adjust the modulation and coding scheme based on channel conditions, which in turn affects the data offloading rate. As a result, resource allocation decisions that appear optimal at one moment may become suboptimal or even infeasible (i.e., deadline miss) by the time task offloading begins. These challenges highlight the necessity for an online adaptive offloading control framework capable of responding to real-time wireless channel quality changes while optimizing resource utilization in VEC. To the best of our knowledge, none of the existing studies on DOP or DOAP has addressed these two critical challenges.

In this paper, we study a DOAP for real-time service subscription in VEC under both bandwidth and computational resource constraints, i.e., the total resources that each RSU can allocate to tasks are bounded by its resource capacities. 
The objective is to maximize the total utility for all vehicles. This paper considers a general and flexible ``task-dependent'' utility function, enabling its applicability to a wide range of real-world use cases. We refer to this problem as $\mathbf{P}$ and formulate it as an integer programming problem. The Generalized Assignment Problem (GAP), a known NP-Hard problem \cite{fleischer2006tight}, can be mapped to an instance of $\mathbf{P}$ where resource allocation is fixed and bandwidth constraint is omitted, implying that $\mathbf{P}$ is also NP-hard. To tackle $\mathbf{P}$, we propose an approximation algorithm leveraging LP rounding and a local ratio technique \cite{bar2004local}, improving the best-known approximation ratio for DOAP from $\frac{1}{6}$ \cite{gao2025local} to $\frac{1}{4}$, while maintaining efficient algorithm-runtime complexity.

To address the challenges of short task deadlines and rapidly changing network conditions in VEC, we propose an online service subscription and offloading control framework for periodic real-time tasks. In this framework, task offloading requests are first queued at the scheduler, which periodically determines service deployment and resource allocation for these requests. Once a task is assigned, its service initialization request is sent to the RSU. After service initialization, an offload grant is issued to the task, allowing the task's job instances to offload (before a grant is received, job instances are processed locally). Since scheduling and service initialization are complete before offloading begins, they do not affect the execution of offloaded job instances. \textit{This mechanism significantly enhances the time-complexity flexibility for scheduling algorithm design, enabling the derivation of finer-grained solutions even for vehicular tasks with short deadlines.}
To adapt to wireless channel fluctuation, vehicles periodically (e.g., every $10$ ms) send Sounding Reference Signals (SRS) to RSUs hosting their services. The RSU then processes the SRS in the 5G module to estimate the uplink channel quality and assesses whether the offloaded jobs under the current channel quality can meet deadlines. If the channel quality degrades below acceptable limits, the RSU temporarily suspends job offloading, requiring jobs to be processed locally in the vehicle. As this assessment is lightweight and the SRS feedback is independent from job offloading, the job execution window is not influenced. This continuous SRS feedback mechanism improves VEC's adaptability to dynamic network conditions. 
Thus, the main contributions of this paper are summarized as follows.
\begin{itemize}
    \item We address a DOAP, $\mathbf{P}$, for real-time service subscription in VEC with both bandwidth and computational resource constraints, aiming to maximize vehicle utility. To solve $\mathbf{P}$, we propose an approximation algorithm, $\mathtt{SARound}$, based on LP rounding and local-ratio techniques, improving the best-known approximation ratio for DOAP from $\frac{1}{6}$ to $\frac{1}{4}$ while preserving algorithm-runtime efficiency.
    \item We introduce an online service subscription and offloading control framework to address the challenges of short task deadlines and rapidly changing wireless channel quality in VEC. To evaluate its effectiveness, we develop a VEC simulator, VecSim, that integrates our framework to manage the full life-cycle of vehicular tasks, including service subscription, resource allocation, and offloading control. Notably, VecSim is the first simulator specifically designed to handle the full life-cycle of deadline-constrained vehicular tasks in VEC.
    \item We evaluate $\mathtt{SARound}$ using profiled object detection workloads and a real-world taxi trace dataset \cite{taxi-data} in VecSim, considering energy saving for vehicles as the utility function. Results show that $\mathtt{SARound}$ consistently outperforms all benchmark algorithms in the literature \cite{hu2020dynamic, dai2018joint, qi2024joint, gao2022deadline, li2019cooperative, fan2023joint, gao2025local} under varying network conditions, and maintains almost linear runtime scalability with respect to the number of requests.
\end{itemize}
The remainder of this paper is organized as follows. Related work is discussed in Section \ref{sec:related-work}. Section \ref{sec:system-architecture} introduces the VEC architecture and problem formulation. The online service subscription and offloading control framework is presented in Section \ref{sec:framework}. Section \ref{sec:algorithm} presents the approximation algorithm $\mathtt{SARound}$. Section \ref{sec:experiment} evaluates $\mathtt{SARound}$'s performance in our developed VEC simulator. Section \ref{sec:conclusion} concludes the paper.  

\section{Related Works}\label{sec:related-work}
Due to the promising prospects of VEC, research interest in task offloading problems within this field has grown significantly. 
Research problems in this area are generally categorized into deadline-free (e.g., response time minimization) and deadline-constrained problems \cite{ramanathan2020a}. This section reviews state-of-the-art research on deadline-constrained problems, which are broadly categorized into two classes: (i) DOP, where only task offloading decisions are made under fixed resource allocations, and (ii) DOAP, where both task offloading and resource allocations are jointly optimized. In both DOP and DOAP, each task can be offloaded to at most one destination out of one or more candidates, and multiple tasks can be offloaded to the same destination as long as the resource constraints are not violated.

Studies on DOP focused solely on task offloading by assuming fixed bandwidth allocations and computational resource allocations being either fixed \cite{lin2023a, li2024computation, liao2024an, ma2022mobility, li2022maximizing, yang2024a, xu2024age} or predetermined based on task deadlines \cite{li2023throughput, zhao2021energy, liu2021energy}. Among these, some studies assumed VEC with neither bandwidth nor computational resource constraint \cite{yang2024a, li2024computation}, some considered VEC with only computational resource constraint \cite{li2022maximizing,  li2023throughput, liao2024an, ma2022mobility, xu2024age, zhao2021energy}, while others investigated VEC scenarios considering both bandwidth and computational resource constraints \cite{lin2023a, liu2021energy}. Specifically, a few studies \cite{zhao2021energy, liu2021energy} considered a single-RSU model, where offloading decisions focused solely on whether to offload a task rather than selecting a target RSU.

In contrast to DOP, DOAP jointly optimizes task offloading and resource allocation \cite{hu2020dynamic, dai2018joint, qi2024joint, vu2021optimal, gao2022deadline, li2019cooperative, fan2023joint, gao2025local}. Several studies assumed a fixed bandwidth allocation for tasks and omitted bandwidth constraints \cite{hu2020dynamic, dai2018joint, qi2024joint, fan2023joint}, proposing heuristic solutions without theoretical guarantees. While some studies incorporated both bandwidth and computational resource constraints, their solutions were either exact with exponential time complexity (e.g., \cite{vu2021optimal}) or heuristic without theoretical guarantees (e.g., \cite{gao2022deadline, li2019cooperative}). 
To balance the runtime efficiency and solution quality, some studies leveraged a recursive local-ratio technique to maximize total system utility \cite{gao2025local}, achieving an approximation ratio of $\frac{1}{6}$. We adapted these algorithms for our DOAP problem, $\mathbf{P}$, and used them as baselines in our experiments. Notably, none of these studies have considered the impact of scheduling overhead and service initialization on offloaded task execution, nor addressed the rapidly changing wireless network quality in VEC. In contrast, this paper addresses a DOAP problem in VEC with bandwidth and computational resource constraints, presenting an approximation algorithm, $\mathtt{SARound}$, that improves the best-known approximation guarantee for DOAP from $\frac{1}{6}$ to $\frac{1}{4}$ while preserving runtime efficiency.
Additionally, we have developed an online service subscription and an adaptive offloading framework to address rapidly changing wireless network quality in VEC.

Building a real testbed for VEC research is technically challenging and costly, making simulators essential tools for evaluating system performance. 
Several widely used edge computing simulators, such as iFogSim \cite{gupta2017ifogsim} and EdgeCloudSim \cite{sonmez2017edgecloudsim}, extend CloudSim to evaluate task scheduling and resource management strategies. However, these simulators lack support for 5G network simulation, making them unsuitable for modeling online offloading control in VEC, where real-time wireless channel quality feedback is essential. On the other hand, network-based simulators such as Simu5G \cite{nardini2020simu5g} (built on OMNeT++ \cite{Varga2010}) and Fogbed \cite{mendonca2021fogbed} (based on Mininet) primarily focus on network characteristics and containerized service deployment. None of these tools provides an integrated computation offloading model that jointly optimizes bandwidth and computational resource allocation for real-time applications or supports online offloading control in dynamic VEC environments. In this paper, we develop VecSim, a VEC simulator that integrates real-time vehicle mobility modeling, 5G-based V2X networking (via Simu5G), and a comprehensive computation offloading mechanism. Unlike existing simulators, VecSim supports online offloading control, joint bandwidth and computational resource allocation for real-time tasks, and task life-cycle management.

\section{System Model and Problem Formulation}\label{sec:system-architecture}
A VEC system comprises a set of RSUs and a scheduler (Fig. \ref{fig:vec_system1}). RSUs are deployed along roadsides to provide services to vehicles, leveraging the 5G network for task offloading due to its high reliability and low latency. Each RSU contains a 5G module that manages the bandwidth allocation for data offloading and a server that provides computational resources for hosting vehicular services. The scheduler communicates with RSUs via Ethernet cables or dedicated long-range wireless networks (e.g., 4G-LTE), and is responsible for determining service deployment and resource allocation for vehicular requests. This study adopts a centralized architecture to enable coordinated scheduling and global optimization. Although deploying decentralized schedulers can improve system robustness against scheduler failures, scheduler fault tolerance is not the primary focus of this work, and investigating decentralized scheduler deployment is left for future research.

Let $\mathcal{M}$ be the set of $M$ RSUs in the VEC, where $m_k \in \mathcal{M}$ denotes the $k$-th RSU. The computational resources of each RSU are measured in computing units (CUs), a metric determined by the specific server hardware configuration. For example, with NVIDIA's Green Contexts feature \cite{nvidia2025green}, the computational resources of a GPU-enabled RSU server can be partitioned into multiple CUs, each comprising a predefined amount of streaming multiprocessor units. Let $C_k$ be the total CUs available at $m_k$, and let $C = \max\{C_k | m_k \in \mathcal{M}\}$. Orthogonal Frequency Division Multiple Access, a core technology in 5G networks, divides a network channel into multiple orthogonal subcarriers. Every $12$ consecutive subcarriers form a resource block (RB), the smallest unit of bandwidth that can be allocated to users. Let $B_k$ be the total RBs available at RSU $m_k$, and let $B = \max\{B_k | m_k \in \mathcal{M}\}$. 

Let $\mathcal{N}$ denote the set of $N$ real-time tasks in the VEC, where $n_i \in \mathcal{N}$ represents the $i$-th task. Let $v_i$ be the vehicle of task $n_i$. Since these tasks typically operate at a specific frequency, task $n_i$ has a period $\delta_i$, representing the rate at which job instances of $n_i$ are released. $\delta_i$ also serves as the implicit relative deadline of $n_i$ and is measured in seconds. Let $\theta_i$ denote the input data size of $n_i$, measured in Megabytes (MB). Since the output data size is typically much smaller than the input and the downlink bandwidth capacity in 5G is generally greater than the uplink capacity, we omit the negligible result return time in our model, as in \cite{lin2023a}. Let $\mathcal{M}_i$ be the set of RSUs accessible by $v_i$, i.e., RSUs that $n_i$'s service can be deployed on. A notation summary is provided in Table \ref{table:notation}.

\begin{figure}
    \centering
    \includegraphics[width=0.8\linewidth, page=1]{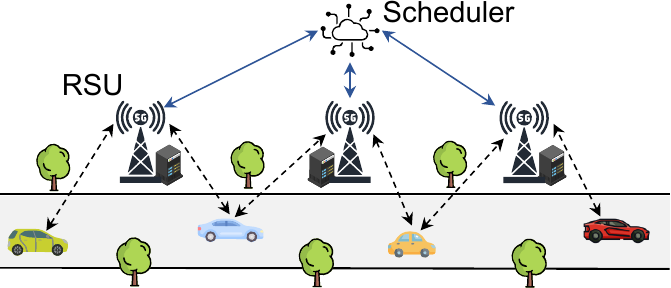}
    \caption{An example of the VEC}
    \label{fig:vec_system1}
\end{figure}


In 5G networks, each vehicle periodically sends Sounding Reference Signals (SRS) to nearby RSUs. Each RSU then processes the received SRS to estimate the wireless channel quality between the vehicle and the RSU \cite{3gpp38211}. The channel quality can be affected by multiple factors, such as the physical distance, data transmission power, signal path loss, obstacles, and interference caused by other vehicles' data transmission. For example, Kutila \emph{et al.} \cite{kutila20205g} have provided an empirical measurement to demonstrate how SRS strength varies due to vehicle mobility. Based on the estimated channel quality, each RSU chooses an appropriate Modulation and Coding Scheme (MCS) for the vehicle to ensure reliable data transmission. The MCS decides the number of bytes that an RB can carry within a fixed time interval, which in turn determines the data offloading rate in wireless networks. Let $\mu_{ik}$ be the data offloading rate per RB per second when task $n_i$ offloads to RSU $m_k$ based on the latest SRS feedback. Mid-band 5G base stations typically provide coverage over a range of $500$ to $1000$ meters in urban environments. Within a $100$ms task period, a vehicle traveling at $60$km/h moves only about $1.67$ meters—a distance negligible compared to the RSU coverage range. Therefore, consistent with prior VEC studies \cite{zhan2020mobility, raza2022task}, we assume that the wireless channel quality (also the $\mu_{ik}$), as measured via uplink SRS, remains approximately constant within a vehicular task period (i.e., of the order of $10$ ms).


\begin{table}
    \caption{Notation (Parameters and Variables)}
    \resizebox{\columnwidth}{!}{
    \begin{tabular}[t]{cl}
        \toprule
        \multicolumn{1}{c}{symb.} & \multicolumn{1}{c}{definition} \\
        \midrule
        $\mathcal{M}$ & set of RSUs, where $|\mathcal{M}| = M$ \\ 
        $C_k$ & total CUs of RSU $m_k$; $C=\max\{C_k|m_k \in \mathcal{M}\}$ \\
        $B_k$ & total RBs of RSU $m_k$; $B=\max\{B_k|m_k \in \mathcal{M}\}$ \\
        $\mathcal{N}$ & set of vehicular tasks, where $|\mathcal{N}|=N$\\
        $v_i$ & the vehicle of task $n_i \in \mathcal{N}$\\
        $\delta_i$ & period (deadline) of task $n_i \in \mathcal{N}$\\
        $\theta_i$ & input data size of task $n_i \in \mathcal{N}$\\
        $\mathcal{M}_i$ & set of RSUs accessible by vehicle $v_i$\\
        $d_i^l$ & local processing time of task $n_i$\\
        $d_{ik}^o$ & offloading time from task $n_i$ to RSU $m_k$\\
        $\mu_{ik}$ & offloading rate per RB per second from $n_i$ to $m_k$\\
        $d_{ik}^p$ & processing time of task $n_i$ on RSU $m_k$\\
        $u_i$ & $u_i(x_{ik},b_{ik},c_{ik})$, utility function of task $n_i$\\
        $\ell$ & $\ell \triangleq \left< n_\ell, m_\ell, b_\ell, c_\ell \right>$, a service instance of task $n_\ell$ \\ 
        $u(\ell)$ & utility value by serving $n_\ell$ following $\ell$\\
        $\mathcal{L}$ & set of all feasible service instances\\
        $\mathcal{L}_i$ & $\mathcal{L}_i \subseteq \mathcal{L}$, set of service instances for task $n_i$ \\
        $\mathcal{L}^k$ & $\mathcal{L}^k \subseteq \mathcal{L}$, set of service instances related to RSU $m_k$ \\
        $\mathcal{L}(\ell)$ & the set of service instances belong to the same task as $\ell$ \\
        $x_{ik}$ & service deployment binary decision variable of $n_i$ \\ 
        $b_{ik}$ & RBs that RSU $m_k$ allocates to task $n_i$ \\ 
        $c_{ik}$ & CUs that RSU $m_k$ allocates to task $n_i$ \\ 
        $z_\ell$ & selection binary variable for service instance $\ell$\\
        \bottomrule
    \end{tabular}}
    \label{table:notation}
\end{table}

\textit{Local Processing.} While not having the permission to offload (either service initialization is not yet done or the offloading is temporarily suspended), the job instances generated by a task are processed locally within the vehicle. When task $n_i$ is processed locally, let $d_i^{l}$ be the local processing time. In this work, we assume $d_i^{l} \le \delta_i$, i.e., job instances can always be scheduled locally within their deadlines. Besides local processing, a task can also be offloaded to RSUs for processing to improve task performance, enhance operational safety, or reduce energy consumption.

\textit{Remote Processing.} We use a binary variable $x_{ik}$ to represent the service deployment decision for task $n_i$. Specifically, $x_{ik} = 1$ if the service for $n_i$ is deployed on RSU $m_k \in \mathcal{M}_i$; otherwise, $x_{ik} = 0$. Let $b_{ik}$ and $c_{ik}$ denote the RBs and CUs that RSU $m_k$ allocates to $n_i$, respectively, where $b_{ik} \in \{0, 1, ..., B_k\}$ and $c_{ik} \in \{0, 1, ..., C_k\}$. The offloading time of each job instance of $n_i$ is calculated as $d_{ik}^{o} = \theta_i/(b_{ik} \cdot \mu_{ik})$. The processing time for each job instance of $n_i$ on $m_k$ is denoted as $d_{ik}^{p}$, which is a function of the allocated CUs $c_{ik}$, the hardware type of $m_k$, and the requested service type of $n_i$. 
The allocated bandwidth and computational resources must ensure that the deadline constraint of each task $n_i$ can be satisfied, i.e.,
\begin{equation}\label{eq:system2}
    d_{ik}^{o} + d_{ik}^{p} = {\theta_i}/{(b_{ik} \cdot \mu_{ik})} + d_{ik}^{p} \le \delta_i.
\end{equation}
Each task $n_i$ is associated with a non-negative and task-dependent utility function $u_i$, whose value $u_i(x_{ik}, b_{ik}, c_{ik})$ depends on $x_{ik}$, $b_{ik}$, and $c_{ik}$. The function $u_i$ generalizes any function that meet the following form:
\begin{enumerate}[(i)]
    \item $u_i(x_{ik}, b_{ik}, c_{ik}) = 0$ if $x_{ik}=0$;
    \item $u_i(x_{ik}, b_{ik}, c_{ik}) = 0$ if the deadline constraint of task $n_i$ is not met, i.e., $d_{ik}^{o} + d_{ik}^{p} > \delta_i$.
\end{enumerate}
This general function allows the system to accommodate diverse performance requirements, such as improved information freshness, enhanced task performance, reduced energy consumption, and
increased operational safety. For instance, when deploying AI applications on resource-constrained onboard devices, model compression techniques (e.g., quantization) are commonly used to reduce computational and memory demands, often at the expense of model accuracy. In such cases, vehicles may prefer to offload tasks to RSUs to obtain more accurate inference results. In this context, $u_i$ can represent the improvement in inference accuracy resulting from offloading. Alternatively, when the system aims to optimize energy consumption, $u_i$ can reflect the energy savings gained by offloading tasks to RSUs.


In VEC, RSUs have limited bandwidth and computational resources, necessitating efficient scheduling decisions for service deployment and resource allocation to maximize VEC resource utilization efficiency. In this paper, our goal is to determine the optimal service deployment and resource allocation $\left<x_{ik}, b_{ik}, c_{ik}\right>$ for each task $n_i$ to maximize the total utility for vehicles while adhering to the resource constraints of RSUs and the task deadline requirements. We refer to this problem as $\mathbf{P}$ and formulate it as follows.
\begin{subequations}
    \allowdisplaybreaks
    \begin{align}
        (\mathbf{P}) \ \ \max \sum_{n_i \in \mathcal{N}}\sum_{m_k \in \mathcal{M}} u_i(x_{ik}, b_{ik}, c_{ik}) & \label{eq:p1}\\
        \text{subject to: }  x_{ik} \cdot (d_{ik}^{o} + d_{ik}^{p}) \le \delta_i, \forall n_i \in \mathcal{N}, & \forall m_k \in \mathcal{M} \label{eq:p1a} \\
        \sum\nolimits_{n_i \in \mathcal{N}} b_{ik} \le B_k,& \forall m_k \in \mathcal{M} \label{eq:p1b} \\
        \sum\nolimits_{n_i \in \mathcal{N}} c_{ik} \le C_k,& \forall m_k \in \mathcal{M} \label{eq:p1c} \\
        \sum\nolimits_{m_k \in \mathcal{M}_i}x_{ik} \le 1, \sum\nolimits_{m_k \in \mathcal{M}\setminus \mathcal{M}_i} x_{ik} = 0,& \forall n_i \in \mathcal{N} \label{eq:p1d} \\
        x_{ik} \in \{0,1\}, b_{ik}, c_{ik} \in \mathbb{Z}_{\ge 0}, \forall n_i \in \mathcal{N}, & \forall m_k \in \mathcal{M} \label{eq:p1f}
    \end{align}
\end{subequations}
Eq. \eqref{eq:p1a} ensures the deadline constraint can be met for each offloaded task. Eqs. \eqref{eq:p1b} and \eqref{eq:p1c} guarantee the bandwidth and computational resource constraints of each RSU will not be violated. Finally, Eq. \eqref{eq:p1d} ensures that a task's service can be deployed on at most one RSU accessible by its vehicle.


\textbf{Problem Complexity.} Given a solution of $\mathbf{P}$, we can easily verify in polynomial time if the solution satisfies all the constraints, making the decision version of $\mathbf{P}$ an NP problem. Besides, a known NP-Hard problem, GAP \cite{fleischer2006tight}, can be mapped to an instance of $\mathbf{P}$ in which the bandwidth and computational resource allocations are predetermined and bandwidth contention is omitted. Because a simplified instance of $\mathbf{P}$ is NP-Hard, $\mathbf{P}$ is also NP-Hard.

\begin{figure}
    \centering
    \includegraphics[width=\linewidth, page=1]{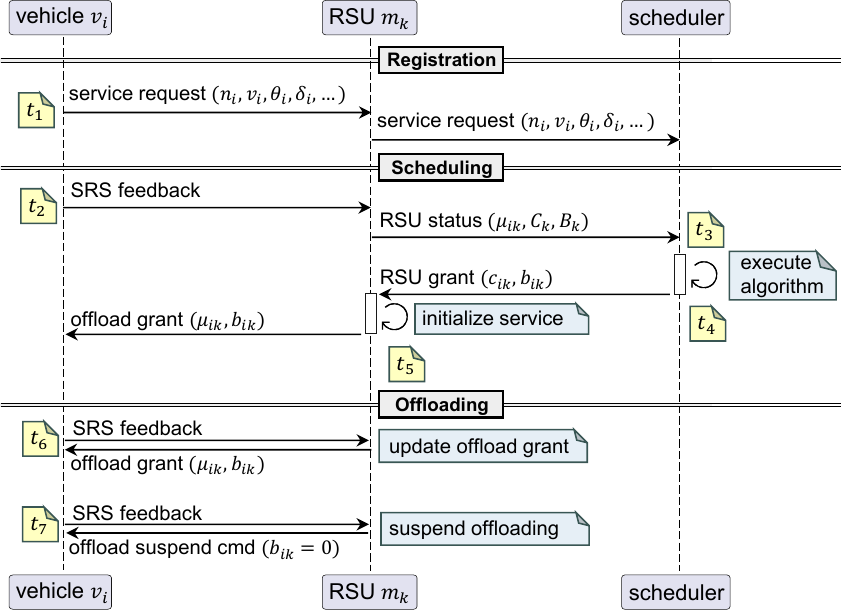}
    \caption{Workflow of the online service subscription and offloading control framework (RSU $m_k$ and task $n_i$ are used as an example).
    }
    \label{fig:framework1}
\end{figure}

In $\mathbf{P}$, resource allocation is determined based on the data offloading rate estimated from the latest SRS feedback before scheduling. However, wireless channel quality changes dynamically in VEC. As a result, resource allocations that are feasible during scheduling may become infeasible by the time task offloading begins. Additionally, an RSU that is accessible to a task during scheduling may subsequently become inaccessible. These challenges underscore the need for online offloading control and for solving $\mathbf{P}$ periodically to adapt service deployments and resource allocations in VEC. Next, we introduce the online offloading control framework in Section \ref{sec:framework}, followed by our solution for $\mathbf{P}$ in Section \ref{sec:algorithm}.

\section{Online Service Subscription and Offloading Control Framework} \label{sec:framework}
To meet the stringent deadline requirements of real-time tasks and address the dynamic network conditions in VEC, this section presents a framework for online service subscription and offloading control. This framework defines the complete life-cycle of these tasks as vehicles move through the VEC coverage area, including service request registration, request scheduling by solving $\mathbf{P}$, and offloading control. A detailed workflow of this framework is illustrated in Fig. \ref{fig:framework1}.

\subsection{Service Request Registration and Scheduling}
\textit{Service Request Registration.} When the vehicle $v_i$ of a task $n_i$ enters the VEC coverage area ($t_1$ in Fig. \ref{fig:framework1}), $v_i$ sends a service request for $n_i$ to its closest RSU, which then forwards the request to the scheduler. The service request includes meta-information about task $n_i$, such as data size, period, required service type, and utility function. Upon receiving the request, the scheduler adds it to the request pool, where it awaits processing during the next scheduling cycle; the request is removed from the pool when the vehicle leaves the VEC coverage area. Until the vehicle receives an offload grant from an RSU, job instances of the task are executed locally.

\textit{RSU Status Reporting.} Before a scheduling cycle starts, each vehicle sends SRS to all nearby RSUs ($t_2$ in Fig. \ref{fig:framework1}). Upon receiving the SRS, the 5G module of an RSU ($m_k$) selects an appropriate MCS to ensure reliable wireless data transmission. Based on the chosen MCS, the 5G module determines the data offloading rate between the vehicle and the RSU ($\mu_{ik}$) and forwards this information to the RSU server. The RSU server then appends the current server resource usage information ($C_k, B_k$) and transmits it to the scheduler. If the data offloading rate between a vehicle $v_i$ and an RSU $m_k$ is not updated to the scheduler before scheduling begins, $m_k$ is considered inaccessible to $v_i$, i.e., $v_i \notin \mathcal{M}_i$.

\textit{Request Scheduling.} The scheduler executes the scheduling algorithm periodically to determine service deployment and resource allocation for requests waiting in the request pool ($t_3$ in Fig. \ref{fig:framework1}). Note that the scheduler is only invoked once at the beginning of each scheduling cycle. After scheduling, an RSU grant (including service deployment and resource allocation) for each admitted task is sent to the corresponding RSU ($t_4$ in Fig. \ref{fig:framework1}). Upon receiving an RSU grant, the RSU $m_k$ initializes the service on its server. Once the service is initialized, the RSU grant is forwarded to the RSU’s 5G module, which allocates the RBs to the corresponding task and issues an offload grant to the respective vehicle ($t_5$ in Fig. \ref{fig:framework1}). 
In our framework, job instances of a task only start offloading when an offload grant is received (i.e., the service is ready on an RSU) and the grant remains alive, \textit{eliminating the impact of scheduling algorithm execution and service initialization on job instance execution window}.

\textit{Scheduling Mode.} Our framework supports two different scheduling modes. In the first mode, all running services on RSUs are terminated before the next scheduling cycle begins, and the scheduler reschedules all the service requests using the full capacity of RSUs; we refer to this mode as \textit{SchedAll}. In the second mode, only unscheduled service requests are scheduled using the remaining RSU resources; we refer to this mode as \textit{SchedRemain}. 
Depending on the VEC requirement, either mode—or even a hybrid approach—can be applied.



\subsection{Online Offloading Control}
In VEC, vehicle movement and environmental changes lead to dynamically changing wireless channel quality, making it challenging to consistently meet job instance deadlines under predetermined resource allocations. To address this issue, we propose an online offloading control protocol for individual job instances based on real-time SRS feedback. Since variations in service execution time on servers are significantly smaller than fluctuations in wireless channel quality, our control protocol primarily focuses on ensuring that job offloading can be completed within its maximum allowable offloading time (i.e., job deadline minus service execution time). The details of the online offloading control protocol are shown in Fig. \ref{fig:framework2}.

Each time a vehicle sends SRS to an RSU, the RSU’s 5G module selects an appropriate MCS to ensure reliable wireless data transmission. Based on the selected MCS, the data offloading rate (per RB per second) between the vehicle and the RSU is computed. Upon receiving an RSU grant after service initialization ($t_5$ in Fig. \ref{fig:framework1}), the 5G module of an RSU verifies whether the \textit{offloading requirement} is satisfied, i.e., whether the job can be offloaded within its maximum allowable offloading time given the current data offloading rate and allocated RBs. If the requirement is met, an offload grant is sent to the vehicle; otherwise, the 5G module suspends the offloading and waits for the next SRS. By configuring the SRS feedback interval to be shorter than the task period (e.g., $1\sim10$ ms), the system can ensure timely updates of channel quality in response to the dynamic nature of VEC networks.

\begin{figure}
    \centering
    \includegraphics[width=\linewidth, page=2]{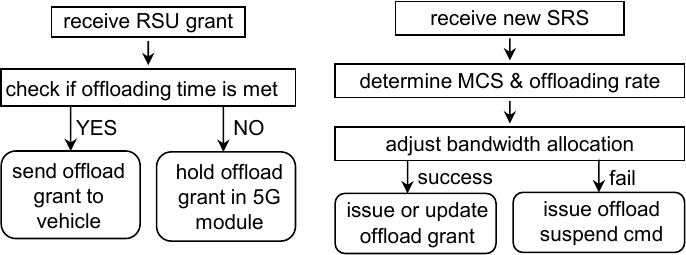}
    \caption{Online offloading control in RSU's 5G module
    }
    \label{fig:framework2}
\end{figure}

During the job offloading phase, each vehicle periodically sends SRS to each RSU hosting its services. Upon receiving new SRS ($t_6$ in Fig. \ref{fig:framework1}), the RSU’s 5G module determines the updated MCS and data offloading rate. For each previously issued offload grant, the 5G module verifies whether the offloading requirement can still be met with the current RB allocation and the new data offloading rate (Fig. \ref{fig:framework2}). If so, no adjustments are needed. Otherwise, the 5G module increases the allocated RBs by assigning the unallocated RBs to the task until the offloading requirement is met. If the available unallocated RBs are insufficient, bandwidth adjustment fails; in this case, the 5G module pauses the grant by sending an offload suspend command to the vehicle ($t_7$ in Fig. \ref{fig:framework1}) and waits for the next SRS feedback. Otherwise, an offload grant is sent to the vehicle to resume offloading or to update its RB allocation. Whenever a job instance of a task is generated, the vehicle offloads the job if an offload grant has been received and is not suspended. \textit{This mechanism ensures the job execution window is not affected by SRS feedback overhead.} 

In this framework, local execution must remain feasible for each task, as RSUs may not always be accessible. Before initiating a new real-time task on a vehicle, the feasibility of local execution must be verified. If local execution is not feasible for the new task, the task is either put on hold or the periods of currently running tasks are adjusted. Moreover, to account for potential processing failures, even when offloading succeeds, this framework focuses on two types of time-critical tasks. The first type includes safety-critical applications (e.g., anomaly detection) that require strict deadline compliance. For these tasks, offloading and local execution are performed concurrently. Local execution acts as a safety fallback in case offloading fails, while successful offloading improves accuracy or responsiveness. In this context, the utility function may reflect gains in accuracy or safety margins. The second type permits occasional deadline violations, modeled using an $M$$-$$K$ constraint, where at least $M$ out of every $K$ consecutive jobs of a task must meet their deadlines. In this case, a job instance is offloaded only if the $M$$-$$K$ condition has already been satisfied.


\begin{figure}
    \centering
    \includegraphics[width=\linewidth, page=3]{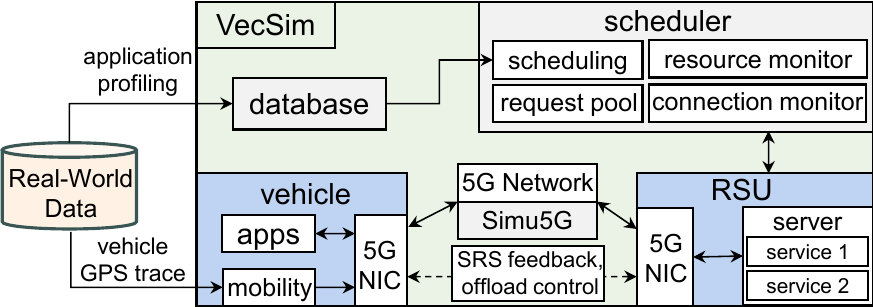}
    \caption{The architecture of the VEC simulator, VecSim.
    }
    \label{fig:framework3}
\end{figure}

\subsection{Framework Implementation}
Building a real testbed for a 5G-enabled VEC system is both technically challenging and costly, making simulation an essential tool for VEC research. Existing simulators lack support for joint bandwidth and computational resource management for real-time applications and online offloading control. To address these issues, we develop a VEC simulator, VecSim\footnote{The source code is available at \url{https://github.com/gaochuanchao/mecRT}}.


The architecture of VecSim is illustrated in Fig. \ref{fig:framework3}. VecSim requires two primary inputs: vehicle GPS traces and application profiling data. GPS traces can be either generated using SUMO \cite{SUMO2018} or collected from real-world vehicle trajectories. These traces are used by VecSim’s mobility module to control vehicle movement within the simulation.
The application profiling data is stored in the database module, which is accessed by the scheduler during resource scheduling.
For network simulation, VecSim leverages the 5G Network Interface Card (NIC) library from Simu5G to manage data transmissions and enable online channel quality estimation in 5G networks. Additionally, we extend the original NIC modules to implement an online offloading control protocol. The RSU server manages computational resources in the RSU and is responsible for service initialization and termination based on the commands from the scheduler. The scheduler module plays a central role in resource orchestration. It maintains (i) a request pool that stores service requests from vehicles within the VEC coverage area, (ii) a resource monitor that tracks the available resources of RSUs (iii) a connection monitor that records the latest data offloading rates between vehicles and RSUs, and (iv) a scheduling module that periodically allocate resource for pending service requests. Since VecSim is built on top of OMNeT++, it benefits from a graphical simulation interface, allowing visualization of vehicle movements, data transmissions, and resource usage during VEC simulation.

\section{Separate Assignment Rounding Algorithm}\label{sec:algorithm}
This section presents the Separate Assignment Rounding Algorithm ($\mathtt{SARound}$) for problem $\mathbf{P}$ formulated in Section \ref{sec:system-architecture}. $\mathtt{SARound}$ is executed periodically in the framework presented in Section~\ref{sec:framework} to allocate RSU resources to pending service requests. It leverages a recursive framework, where each recursive layer employs an LP rounding method to determine the service deployment and resource allocation for exactly one RSU. Later in this section, we formally prove that $\mathtt{SARound}$ is a $\frac{1}{4}$-approximation algorithm for $\mathbf{P}$. 

\subsection{ILP Formulation}
\begin{definition}[Service Instance] \label{definition:instance}
    A service instance of task $n_i$ is defined as $\ell \triangleq \left< n_\ell, m_\ell, b_\ell, c_\ell \right>$, representing that the service for task $n_\ell$ ($n_\ell = n_i$) is deployed on RSU $m_\ell$, and is allocated $b_\ell$ RBs for task offloading and $c_\ell$ CUs for task processing. Moreover, $\ell$ needs to satisfy (i) the RSU access constraint (i.e., $m_\ell \in \mathcal{M}_i$), (ii) the resource constraint (i.e., $b_\ell \le B_\ell$ and $c_\ell \le C_\ell$), and (iii) the deadline constraint (i.e., Eq. \eqref{eq:system2}).
\end{definition}

\textit{Solving $\mathbf{P}$ directly is challenging due to the generalized utility function $u_i$ and the nonlinear deadline constraint \eqref{eq:p1a}.} Therefore, we first reformulate $\mathbf{P}$ as an equivalent Integer Linear Programming (ILP) problem, denoted as $\mathbf{P_{ILP}}$. Let $\ell$ be a service instance of task $n_i$ (as in Definition \ref{definition:instance}). Suppose $m_\ell=m_k$, i.e., $x_{ik}=1$. Given $x_{ik}$, $b_\ell$, $c_\ell$, the corresponding utility value $u_i(x_{ik},b_\ell, c_\ell)$ can be determined. For simplicity, we denote this utility value as $u(\ell)$.
Since resource allocations are discrete, we enumerate all feasible service instances for each task. Let $\mathcal{L}$ be the set of all feasible service instances for all tasks in $\mathcal{N}$, $\mathcal{L}_i \subseteq \mathcal{L}$ be the set of all service instances for task $n_i$, and $\mathcal{L}^k \subseteq \mathcal{L}$ be the set of all service instances related to RSU $m_k$. For any service instance $\ell$ and task $n_{\ell}$, we use $\mathcal{L}(\ell)$ to denote the set of all service instances belonging to the same task $n_\ell$ (i.e., when $n_{\ell}=n_i$, $\mathcal{L}(\ell)=\mathcal{L}_i$). The VEC comprises $N$ tasks and $M$ RSUs, and each RSU has at most $B$ RBs and $C$ CUs. Therefore, $\mathcal{L}$ contains at most $NMBC$ service instances. For each service instance $\ell \in \mathcal{L}$, let $z_\ell \in \{0,1\}$ be the selection variable of $\ell$, where $z_\ell = 1$ if and only if $\ell$ is selected as the service instance for task $n_\ell$ in the solution. Then, we reformulate $\mathbf{P}$ as an ILP problem $\mathbf{P_{ILP}}$ as follows.
\begin{subequations}
    \allowdisplaybreaks
    \begin{align}
        (\mathbf{P_{ILP}}) \ \ \max \textstyle \sum\nolimits_{\ell \in \mathcal{L}} & u(\ell) \cdot z_\ell \label{eq:p2} \\
        \text{subject to:} \hspace{0.5cm} \textstyle \sum\nolimits_{\ell \in \mathcal{L}^k} b_\ell \cdot z_\ell \le B_k,& \ \forall m_k \in \mathcal{M} \label{eq:p2a} \\
        \textstyle \sum\nolimits_{\ell \in \mathcal{L}^k} c_\ell \cdot z_\ell \le C_k,& \ \forall m_k \in \mathcal{M} \label{eq:p2b} \\
        \textstyle \sum\nolimits_{\ell \in \mathcal{L}_i} z_\ell \le 1,& \ \forall n_i \in \mathcal{N} \label{eq:p2c} \\
        z_\ell \in \{0,1\},& \ \forall \ell \in \mathcal{L} \label{eq:p2d}
    \end{align}
\end{subequations}
Eqs. \eqref{eq:p2a} and \eqref{eq:p2b} are the resource constraints. Eq. \eqref{eq:p2c} ensures at most one service instance can be selected for each task. Note that the deadline constraints in Eq. \eqref{eq:p1a} and RSU access constraints in Eq. \eqref{eq:p1d} of $\mathbf{P}$ have been implicitly satisfied via the definition of service instances (Definition \ref{definition:instance}). 
As an ILP problem, solving $\mathbf{P_{ILP}}$ optimally requires exponential time, which is impractical for dynamic VEC systems. To balance solution quality and runtime efficiency, we propose an approximation algorithm for $\mathbf{P_{ILP}}$ in Subsection \ref{subsec:saround}, which yields a provably suboptimal solution with polynomial runtime complexity.

\subsection{Separate Assignment Rounding Algorithm (SARound)}\label{subsec:saround}
This subsection introduces our proposed approximation algorithm $\mathtt{SARound}$ (Algorithm \ref{alg:sa-round}). $\mathtt{SARound}$ is a recursive algorithm, initialized with the call $\mathtt{SARound}$($1$, $\mathbf{u}$, $\mathcal{L}$), where $\mathbf{u}$ is a vector of $u(\ell)$ for all $\ell \in \mathcal{L}$. At each recursive layer, $\mathtt{SARound}$ updates the utility value of remaining service instances.
For ease of presentation, we denote the \textit{updated utility value} of service instance $\ell$ in the $k$-th layer as $w^k(\ell)$, with $u(\ell)$ representing the \textit{original utility value} of $\ell$ (i.e., $u_i(x_{ik},b_\ell, c_\ell)$). Let $\mathbf{w^k}$ be the vector of $w^k(\ell)$ for all $\ell$. For any service instance set $\mathcal{S}$, we define $w^k(\mathcal{S}) = \sum_{\ell \in \mathcal{S}} w^k(\ell)$.

Let $\mathbf{P_{ILP}^k}$ represent the ILP subproblem of $\mathbf{P_{ILP}}$ when restricted to a single RSU $m_k$ and the utility vector $\mathbf{w^k}$. In the $k$-th recursive layer of $\mathtt{SARound}$, we invoke function $\mathtt{FloorRd}$ (lines $7$--$11$) to determine the solution for $\mathbf{P_{ILP}^k}$ (i.e., the service deployment on RSU $m_k$). As $\mathbf{P_{ILP}^k}$ is still an integer problem, solving it optimally takes exponential time. Therefore, in $\mathtt{FloorRd}$, we apply an LP rounding technique to obtain a sub-optimal solution for $\mathbf{P_{ILP}^k}$ with a more practical time complexity. We first relax $\mathbf{P_{ILP}^k}$ into an LP problem, denoted as $\mathbf{P_{LP}^k}$, by relaxing the binary variable $z_\ell$ into a continuous variable of range $[0,1]$. Given the utility vector $\mathbf{w^k}$, the relaxed LP problem $\mathbf{P_{LP}^k}$ for single RSU $m_k$ can be defined as follows.
\begin{subequations}
    \allowdisplaybreaks
    \begin{align}
        (\mathbf{P_{LP}^k}) \ \ \max \textstyle \sum\nolimits_{\ell \in \mathcal{L}^k} w^k(\ell) \cdot z_\ell & \label{eq:rp2m} \\
        \text{subject to:} \hspace{0.5cm} \textstyle \sum\nolimits_{\ell \in \mathcal{L}^k} b_\ell \cdot z_\ell \le B_k & \label{eq:rp2ma} \\
        \textstyle \sum\nolimits_{\ell \in \mathcal{L}^k} c_\ell \cdot z_\ell \le C_k & \label{eq:rp2mb} \\
        \textstyle \sum\nolimits_{\ell \in \mathcal{L}_i\cap\mathcal{L}^k} z_\ell \le 1, \forall n_i \in \mathcal{N} & \label{eq:rp2mc} \\
        z_\ell \ge 0, \forall \ell \in \mathcal{L}^k & \label{eq:rp2md}
    \end{align}
\end{subequations}
Note the upper limit for $z_\ell$ (i.e., 1) is specified by Eq. \eqref{eq:rp2mc}. The ILP problem $\mathbf{P_{ILP}^k}$ has the same formulation as $\mathbf{P_{LP}^k}$ except that the feasible domain of $z_\ell$ is $\{0,1\}$.
$\mathtt{FloorRd}$ first computes (e.g., using the simplex method) an \textit{optimal basic solution}, $\mathbf{z}^*$, for the LP problem $\mathbf{P_{LP}^k}$ (line $8$). Note, solving an LP problem with $n$ variables has time complexity $\mathcal{O}(n^3)$ \cite{vaidya1987an}. Because of potential fractions (i.e., $0 < z_\ell^* < 1$ for some $\ell \in \mathcal{L}^k$) in the solution for $\mathbf{P_{LP}^k}$, it may not be a feasible solution for $\mathbf{P_{ILP}^k}$. Therefore, we \textit{round down} all fractional $z_\ell^*$ to obtain a fully integral solution $\mathcal{I}$ (line $9$), i.e., all resulting $z_\ell^*$ after rounding is either $0$ or $1$. As constraints \eqref{eq:rp2ma} to \eqref{eq:rp2mc} are still satisfied after removing fractional $z_\ell^*$, $\mathcal{I}$ is a feasible solution for $\mathbf{P_{ILP}^k}$. Then, $\mathcal{I}$ is compared with the service instance $\ell \in \mathcal{L}^k$ having the largest utility $w^k(\ell)$, and the one with higher total utility is returned as the solution for RSU $m_k$ (lines $10$--$11$). This comparison is to ensure that the solution quality is lower bounded when the loss due to rounding is large.

\begin{algorithm}[tb]
    \caption{$\mathtt{SARound}$}\label{alg:sa-round}
    
    \SetKwFunction{Round}{$\mathtt{FloorRd}$}
    \SetKwFunction{SARound}{$\mathtt{SARound}$}
    \SetKwProg{Fn}{Function}{:}{}
    \SetKwProg{Main}{}{:}{}
    \SetKwInOut{Input}{input}
    \SetKwInOut{Output}{output}
    \SetKw{Return}{return}

    \Main{\SARound{$k$, $\mathbf{w^k}$, $\mathcal{L}$}}{
    $\mathcal{F}^k \leftarrow$ \Round{$\mathbf{w^k}$, $\mathcal{L}^k$}\;
    Decompose utility vector $\mathbf{w^k} = \mathbf{w_1^k} + \mathbf{w_2^k}$ such that
    \begin{equation*}
        \forall \ell' \in \mathcal{L}, w_1^k(\ell') = 
        \begin{cases}
            w^k(\ell') & \text{if } \ell' \in \mathcal{L}^k\\
            w^k(\ell) & \text{if }  \scaleto{\exists \ell \in \mathcal{F}^k, \ell' \in \mathcal{L}(\ell)\setminus \mathcal{L}^k}{10pt}\\
            0 & \text{otherwise}
        \end{cases}
    \end{equation*}
    
    \textbf{if} $k < M$ \textbf{then} $\mathcal{S}^{k+1} \leftarrow $ \SARound{$k+1$, $\mathbf{w_2^k}$, $\mathcal{L}$} \textbf{else} $\mathcal{S}^{k+1} \leftarrow \emptyset$\;
    Remove any $\ell$ from $\mathcal{F}^k$ if exist $\ell' \in \mathcal{L}(\ell) \cap \mathcal{S}^{k+1}$\;
    \Return $\mathcal{S}^{k} \leftarrow \mathcal{S}^{k+1} \cup \mathcal{F}^k$.
    }
    \BlankLine
    \Fn{\Round{$\mathbf{w^k}, \mathcal{L}^k$}}{
    Let $\mathbf{z}^*$ be an \textit{optimal basic solution} of $\mathbf{P_{LP}^k}$\;
    $z_\ell \leftarrow \lfloor z^*_\ell \rfloor$ for all $\ell \in \mathcal{L}^k$; $\mathcal{I} \leftarrow \{\ell \in \mathcal{L}^k | z_\ell = 1\}$\;
    Let $\ell_{max}$ be the $\ell \in \mathcal{L}^k$ with a largest $w^k(\ell)$)\;
    \scalebox{0.97}{\leIf{$w^k(\mathcal{I}) < w^k(\ell_{max})$}{\Return $\ell_{max}$}{\Return $\mathcal{I}$}}
    }
\end{algorithm}

At each recursive layer of $\mathtt{SARound}$, the solution obtained for RSU $m_k$ by $\mathtt{FloorRd}$ (i.e., $\mathcal{F}^k$) is used to update the utility vector $\mathbf{w^k}$ (line $3$) for service instances that can be affected (i.e., related to the same RSU $m_k$ or belong to the same task as instances in $\mathcal{F}^k$). Specifically, the utility $w^k(\ell)$ of each service instance $\ell$ is decomposed into two parts, $w_1^k(\ell)$ and $w_2^k(\ell)$. The vector $\mathbf{w_2^k}$ can be interpreted as the ``marginal gain'' of the remaining service instances corresponding to the selected service instances in $\mathcal{F}^k$, which is used as the utility vector $\mathbf{w^{k+1}}$ for the next recursive layer. Let $\mathcal{S}^{k+1}$ be the solution returned from the next recursive layer (line $4$). To satisfy constraint \eqref{eq:p2c}, an instance $\ell$ in $\mathcal{F}^k$ is excluded if any other service instance $\ell'$ belonging to the same task as $\ell$ (i.e., $\ell' \in \mathcal{L}(\ell)$) already exists in  $\mathcal{S}^{k+1}$. The returned $\mathcal{S}^k$ is then defined as the union of the updated $\mathcal{F}^k$ and $\mathcal{S}^{k+1}$ (lines $5$–$6$).

\textbf{Time Complexity.} In each call to $\mathtt{FloorRd}$, there are at most $NBC$ service instances in $\mathcal{L}^k$. The time complexity of determining $\ell_{max}$ is $\mathcal{O}(NBC)$. The time complexity for solving the LP problem $\mathbf{P_{LP}^k}$ with $NBC$ variables is $\mathcal{O}((NBC)^3)$. Therefore, at each recursive layer of $\mathtt{SARound}$, the time complexity of calling $\mathtt{FloorRd}$ (line $2$) is $\mathcal{O}((NBC)^3)$, while decomposing the utility values (line $3$) requires $\mathcal{O}(MNBC)$. Since $\mathtt{SARound}$ has $M$ recursive layers, its overall time complexity is $\mathcal{O}(M^2(NBC)^3)$, which is polynomial in the input size of $\mathbf{P_{ILP}}$. However, because the complexity depends on the values of $B$ and $C$, $\mathtt{SARound}$ is a pseudo-polynomial-time algorithm for $\mathbf{P}$. In practice, though, $B$ and $C$ are typically small. For example, a 5G communication channel generally provides up to $275$ RBs~\cite{ali20185g}, and an RTX 4090 GPU can be partitioned into at most $32$ instances using Nvidia's Green Contexts feature~\cite{nvidia2025green}. Therefore, $B$ and $C$ can be regarded as constants in real-world systems, making $\mathtt{SARound}$ effectively a polynomial-time algorithm for $\mathbf{P}$.

In the remainder of this subsection, we first show that $\mathtt{FloorRd}$ is a $\frac{1}{3}$-approximation algorithm for $\mathbf{P_{ILP}^k}$ when restricted to a single RSU $m_k$ and utility vector $\mathbf{w^k}$. We then prove that $\mathtt{SARound}$ achieves a $\frac{1}{4}$-approximation for the original problem $\mathbf{P_{ILP}}$, or equivalently $\mathbf{P}$.

\begin{theorem}\label{theorem:floorRd}
    Function $\mathtt{FloorRd}$ is a $\frac{1}{3}$-approximation algorithm for the ILP problem $\mathbf{P_{ILP}^k}$.
\end{theorem}
\begin{proof}
    $\mathtt{FloorRd}$ first solves the relaxed LP,  $\mathbf{P_{LP}^k}$, and then rounds the fractional solution to derive an integral solution for $\mathbf{P_{ILP}^k}$. We prove this theorem by showing that the loss due to rounding can be bounded.
    In a \textit{basic feasible solution} of an LP problem with $n$ functional constraints (excluding non-negative constraints), there can be at most $n$ variables taking positive values, termed basic variables, and the rest of the variables take value $0$, termed non-basic variables \cite{hillier2015introduction}. The LP problem $\mathbf{P_{LP}^k}$ comprise $(N+2)$ functional constraints ($2$ for constraints \eqref{eq:rp2ma} and \eqref{eq:rp2mb}, and $N$ for constraint \eqref{eq:rp2mc}). Thus, \textit{an optimal basic solution, which is also a basic feasible solution, of $\mathbf{P_{LP}^k}$ has at most $(N+2)$ basic variables that take positive values}.
    
    
    The slack form of each equation in constraint \eqref{eq:rp2mc} is given by $\sum_{\ell \in \mathcal{L}_i\cap \mathcal{L}^k} z_\ell + s_i = 1, \forall n_i \in \mathcal{N}$, where $s_i$ is the slack variable. Note that the slack variable can be a basic variable if it takes a positive value in a basic feasible solution, even though it contributes $0$ to the objective function. In each of these $N$ equations, the left-hand side sums to $1$; thus, each equation must contain at least one basic variable with a positive value.
    In an optimal basic solution of $\mathbf{P_{LP}^k}$ (i.e., $\mathbf{z}^*$), if a fractional basic variable ($0 < z_\ell^* < 1$) appears in any of these $N$ equations, there must be at least one additional fractional basic variable in the same equation, since the variables in that equation sum to $1$. In the worst-case scenario of maximum rounding loss, there are at most two such equations, each containing two basic variables with positive \textit{fractional} values. This results in at most four basic variables, associated with two tasks, having values in the range $(0,1)$. Otherwise, more than $(N+2)$ variables would take positive values, which violates the properties of a basic feasible solution in LP problems. Thus, the rounding process in line $9$ of $\mathtt{SARound}$ removes at most four service instances of two tasks. Let $\ell_1, \ell_2, \ell_3$, and $\ell_4$ denote these removed service instances; suppose $\ell_1$ and $\ell_2$ belong to task $n_1$ (i.e., $\ell_1, \ell_2 \in \mathcal{L}_1$), and $\ell_3$ and $\ell_4$ belong to task $n_2$ (i.e., $\ell_3, \ell_4 \in \mathcal{L}_2$). Let
    \[\scaleto{W_1 = z^*_{\ell_1}w^k(\ell_1)  + z^*_{\ell_2}w^k(\ell_2) }{11pt} \text{ and } \scaleto{W_2 = z^*_{\ell_3}w^k(\ell_3)  + z^*_{\ell_4}w^k(\ell_4) .}{11pt}\] 
    According to the definition of $\ell_{max}$ in line $10$ of $\mathtt{SARound}$ and constraint \eqref{eq:rp2mc} (i.e., $z^*_{\ell_1} + z^*_{\ell_2} \le 1$ and $z^*_{\ell_3} + z^*_{\ell_4} \le 1$), we have $w^k(\ell_{max}) \ge \max\{W_{1}, W_{2}\}$. Thus, the total utility of the returned solution $\mathcal{F}^k$ (line $2$) can be expressed as:
    \begin{equation}\label{eq:theorem1b}
        \scaleto{w^k(\mathcal{F}^k) = \max\{w^k(\mathcal{I}), w^k(\ell_{max})\} \ge \max\{w^k(\mathcal{I}), W_1, W_2\}.}{10.3pt}
    \end{equation}
    Here, the $\mathcal{I}$ is the set of service instances obtained after rounding in line $9$ of $\mathtt{SARound}$.
    Let $\mathcal{Q}^k$ denote the optimal solution of the ILP problem $\mathbf{P_{ILP}^k}$. Since the optimal objective of $\mathbf{P_{LP}^k}$ serves as an upper bound for that of $\mathbf{P_{ILP}^k}$, we have $w^k(\mathcal{I}) + W_1 + W_2 = \sum_{\ell \in \mathcal{L}^k}w^k(\ell)z_\ell^* \ge w^k(\mathcal{Q}^k)$, which implies that at least one term in $\{w^k(\mathcal{I}), W_1, W_2\}$ is at least $\frac{1}{3}w^k(\mathcal{Q}^k)$; otherwise, the inequality does not hold. Combining this result with Eq. \eqref{eq:theorem1b}, we conclude that: 
    $w^k(\mathcal{F}^k) \ge \frac{1}{3}w^k(\mathcal{Q}^k)$.
\end{proof}

Based on Theorem \ref{theorem:floorRd}, we can derive the approximation ratio of $\mathtt{SARound}$ for problem $\mathbf{P_{ILP}}$ as follows.
\begin{theorem}\label{theorem:saRound}
    $\mathtt{SARound}$ is a $\frac{1}{4}$-approximation algorithm for the problem $\mathbf{P_{ILP}}$ (or equivalently $\mathbf{P}$).
\end{theorem}
\begin{proof}
    Let $\mathcal{Q}$ be an optimal solution for $\mathbf{P_{ILP}}$. We use \textit{simple induction} to prove this theorem by showing that $w^k(\mathcal{S}^k) \ge \frac{1}{4}w^k(\mathcal{Q})$ for $k = M, M-1, ..., 1$, where $\mathbf{w^1} = \mathbf{u}$.

    \textbf{Base Case $(k=M)$:} When $k=M$, $\mathcal{S}^M = \mathcal{F}^M$. From line $3$, for $k=1, ..., M$, we set $w_1^k(\ell') = w^k(\ell')$ for all $\ell' \in \mathcal{L}^k$; thus, $w_2^k(\ell') = 0$ for all $\ell' \in \mathcal{L}^k$ after decomposition. Therefore, when $k=M$, only service instance $\ell \in \mathcal{L}^M$ retain a positive utility value $w^k(\ell)$.
    By Theorem \ref{theorem:floorRd}, we have: 
    \[\textstyle{w^M(\mathcal{S}^M) = w^M(\mathcal{F}^M) \ge \frac{1}{3}w^M(\mathcal{Q}) \ge \frac{1}{4}w^M(\mathcal{Q}).}\]

    \textbf{Inductive Step $(k<M)$:} In this step, we show that for $k < M$, if $w^{k+1}(\mathcal{S}^{k+1}) \ge \frac{1}{4}w^{k+1}(\mathcal{Q})$, then $w^{k}(\mathcal{S}^{k}) \ge \frac{1}{4}w^{k}(\mathcal{Q})$. Assume $w^{k+1}(\mathcal{S}^{k+1}) \ge \frac{1}{4}w^{k+1}(\mathcal{Q})$. In line $4$ of $\mathtt{SARound}$, we pass $\mathbf{w_2^k}$ as the new utility vector in the recursive call (i.e., $\mathtt{SARound}(k+1,\mathbf{w_2^k},\mathcal{L})$). Thus, $\mathbf{w_2^k} = \mathbf{w^{k+1}}$. Based on the assumption of $w^{k+1}(\mathcal{S}^{k+1}) \ge \frac{1}{4}w^{k+1}(\mathcal{Q})$, we have $w_2^{k}(\mathcal{S}^{k+1}) \ge \frac{1}{4}w_2^{k}(\mathcal{Q})$. In line $3$ of $\mathtt{SARound}$, we set $w_1^k(\ell') = w^k(\ell')$ for all $\ell' \in \mathcal{L}^k$. Since $w_2^k(\ell') = w^k(\ell') - w_1^k(\ell')$, it follows that $w_2^k(\ell') = 0$ for all $\ell' \in \mathcal{L}^k$. Thus: 
    \begin{equation}\label{eq:theorem1-1}
        {\textstyle w_2^{k}(\mathcal{S}^{k}) = w_2^{k}(\mathcal{S}^{k+1}) \ge \frac{1}{4}w_2^{k}(\mathcal{Q}).}
    \end{equation}
    In line $5$, we remove any $\ell$ from $\mathcal{F}^k$ if there exists a service instance $\ell' \in \mathcal{S}^{k+1}$ that belongs to the same job (i.e., $\ell' \in \mathcal{L}(\ell)$); let $\mathcal{F}_{rm}^k$ be the set of service instances removed from $\mathcal{F}^k$. We can partition the solution $\mathcal{S}^k$ into a set $\mathcal{S}_1^k = \mathcal{F}^k \setminus \mathcal{F}_{rm}^k$ (the remaining instances in $\mathcal{F}^k$ after removing $\mathcal{F}_{rm}^k$), a set $\mathcal{S}_2^k$ containing the instances causing $\mathcal{F}_{rm}^k$ to be removed from $\mathcal{F}^k$, and a set $\mathcal{S}_3^k$ containing the remaining instances in $\mathcal{S}^k$. Based on the rule of decomposing $\mathbf{w^k}$ in line $3$, we have $w_1^k(\mathcal{S}_1^k) = w^k(\mathcal{S}_1^k) = w^k(\mathcal{F}^k \setminus \mathcal{F}_{rm}^k)$, $w_1^k(\mathcal{S}_2^k) = w^k(\mathcal{F}_{rm}^k)$, and $w_1^k(\mathcal{S}_3^k) = 0$.
    Hence, 
    \[w_1^k(\mathcal{S}^k) = w^k(\mathcal{F}^k \setminus \mathcal{F}_{rm}^k) + w^k(\mathcal{F}_{rm}^k) = w^k(\mathcal{F}^k).\]
    Meanwhile, the optimal solution $\mathcal{Q}$ can also be partitioned into a set $\mathcal{Q}_1$ containing instances in $\mathcal{L}^k$, a set $\mathcal{Q}_2$ containing instances in $\{\ell' \in \mathcal{Q} \mid \forall \ell \in \mathcal{F}^k, \ell' \in \mathcal{L}(\ell)\setminus \mathcal{L}^k\}$ (a similar case as $\mathcal{S}_2^k$), and a set $\mathcal{Q}_3$ containing the remaining instances in $\mathcal{Q}$. Based on line $3$ and Theorem \ref{theorem:floorRd}, we have $w_1^k(\mathcal{Q}_1) = w^k(\mathcal{Q}_1)\le3w^k(\mathcal{F}^k)$, $w_1^k(\mathcal{Q}_2) = w^k(\mathcal{Q}_2) \le w^k(\mathcal{F}^k)$, and $w_1^k(\mathcal{Q}_3) = 0$. Thus: 
    \[w_1^k(\mathcal{Q}) = w_1^k(\mathcal{Q}_1) + w_1^k(\mathcal{Q}_2) + w_1^k(\mathcal{Q}_3) \le 4w^k(\mathcal{F}^k). \]
    Since $w_1^k(\mathcal{S}^k) = w^k(\mathcal{F}^k)$, we have $w_1^k(\mathcal{S}^k) \ge \frac{1}{4}w_1^k(\mathcal{Q})$. Given $\mathbf{w^k}= \mathbf{w_1^k}+\mathbf{w_2^k}$ and Eq. \eqref{eq:theorem1-1}, we can derive $w^k(\mathcal{S}^k) \ge \frac{1}{4}w^k(\mathcal{Q})$. 
    Combining the results of the base case and the inductive step, we can conclude $w^k(\mathcal{S}^k) \ge \frac{1}{4}w^k(\mathcal{Q})$ for $k=M,..., 1$. Since the initial call to $\mathtt{SARound}$ is $\mathtt{SARound}$($1$, $\mathbf{u}$, $\mathcal{L}$), we have $w^1(\mathcal{S}^1) = u(\mathcal{S}^1) \ge \frac{1}{4}u(\mathcal{Q})$.
\end{proof}

\textbf{Discussion.} Based on the classification of knapsack problems \cite{CACCHIANI2022105693}, the DOAP $\mathbf{P}$ (or its equivalent $\mathbf{P_{ILP}}$) is a Two-dimensional Multiple-Choice Multiple Knapsack Problem (2DMCMKP). The best-known approximation guarantee for 2DMCMKP is $\frac{1}{6}$ \cite{gao2025local}, and we improve it to $\frac{1}{4}$ in this work. $\mathtt{SARound}$ leverages the $\mathtt{FloorRd}$ function to compute the solution for each RSU with the following LP property: in a basic feasible solution, the number of variables that can take positive values does not exceed the number of functional constraints. This fundamental property was employed by Patt-Shamir and Rawitz \cite{PATTSHAMIR20121591} for the multidimensional multiple-choice single knapsack problem. If we replace $\mathtt{FloorRd}$ with the method by Patt-Shamir and Rawitz in $\mathtt{SARound}$, we can further improve the theoretical guarantee for DOAP to $\frac{1}{2+\epsilon}$, with time complexity of $\mathcal{O}(M^2(NBC)^{3+\lceil2/\epsilon\rceil})$. However, this approach is computationally prohibitive: even when $\epsilon=2$, achieving the same guarantee as $\mathtt{SARound}$, its time complexity is approximately $NBC$ times higher, making it impractical for real-world scenarios. Furthermore, $\mathtt{SARound}$ is adaptable to $\mathbf{P}$ even when additional resource constraints are introduced, such as server memory constraints.

\section{Experimental Evaluation}\label{sec:experiment}
This section evaluates the performance of $\mathtt{SARound}$ by comparing it with four state-of-the-art benchmark algorithms from the literature. Specifically, we analyze the impact of network quality levels and scheduling modes on algorithm performance. All experiments were conducted on a desktop PC equipped with an Intel(R) Xeon(R) W-2235 CPU and 32 GB of RAM.

This paper considers a general and flexible "task-dependent" utility function, enabling its applicability to a wide range of real-world use cases. In this experimental evaluation, we demonstrate how complex objectives, such as energy saving, can be incorporated into this utility function. It is worth mentioning that, beyond energy saving, the utility function can also represent metrics such as age of information, operational safety, or application performance.  We define the utility function $u_i$ for each task $n_i$ in terms of the energy savings for vehicles, which depends on both the energy saving per job instance of $n_i$ and the frequency at which the job instance is released. Thus, we measure the energy saving of each task as the energy saving per second. Let $p_i^{l}$ denote the local processing power required for task $n_i$, and the energy consumption for local execution can be computed as $e_i^{l} = p_i^{l} \cdot d_i^{l}/\delta_i$. Suppose task $n_i$ is offloaded to RSU $m_k$ for processing. Let $p_i^{o}$ represent the offloading power of vehicle $v_i$. The energy consumption per second for offloading $n_i$ is given by $e_{ik}^{o} = p_i^{o} \cdot d_{ik}^{o}/\delta_i$. Thus, the energy saving per second for vehicle $v_i$ when offloading and executing $n_i$ on RSU $m_k$ is given by $e_{ik}^{save} = e_i^{l} - e_{ik}^{o}$ (the energy consumption of SRS feedback is omitted due to its negligible energy consumption compared with task offloading), i.e.,
\[\textstyle u_i(x_{ik},b_{ik},c_{ik})=e_{ik}^{save}=x_{ik}(p_i^{l} \cdot d_i^{l}-p_i^{o} \cdot \frac{\theta_i}{b_{ik}\cdot \mu_{ik}})/\delta_i.\]

\subsection{Simulation Setup}
Constructing a real VEC testbed is both technically challenging and costly. Therefore, we evaluate algorithm performance using our VEC simulator, VecSim. The simulation environment is built using real taxi GPS trajectory data \cite{taxi-data} to model vehicle mobility, collected in Shanghai on April 1, 2018, by the Shanghai Qiangsheng Taxi Company. We extract all taxi trajectories passing through a $1$ km $\times$ $1$ km urban area in Shanghai over a $15$-minute period in the evening, during which $80$ taxis traverse the area. To provide real-time services to surrounding vehicles, we deploy $15$ RSUs evenly along the roads. Each RSU manages a 5G network with $270$ RBs (i.e., $B_k=270$), delivering a total data rate of $37$ MB/s (obtained from profiling) under optimal channel quality. Additionally, each RSU is equipped with a GPU server containing $16$ CUs (i.e., $C_k=16$), where the GPU type is randomly selected from six models (TITAN V, TITAN RTX, RTX 2080 Ti, RTX 3090, RTX 4060 Ti, A100). Notably, for the same allocated CUs, the execution time of a given service varies across different GPU types, which are obtained through profiling task execution in each type of GPU. Furthermore, the service initialization time on various RSUs is randomly sampled from the range of $[10, 50]$ ms. Given the coverage range of $500$ to $1000$ meters for mid-band 5G base stations in urban environments, and an average vehicle speed of $60$km/h ($16.7$m/s), we set the scheduling interval to $10$ seconds to balance the dynamics of system state and scheduling overhead.

We obtain local execution data for real-time tasks by profiling them on Nvidia Jetson Nano, an embedded device capable of running GPU applications. Each vehicle runs one or more applications (tasks). For each task $n_i$, input data is randomly sampled from a set of 57 images in the ImageNet ILSVRC dataset \cite{ILSVRC15}, with sizes $\theta_j$ ranging from $0.07$MB to $0.3$MB. The application type is also randomly selected from six GPU-based object detection models (RESNET-18/50/101/152, VGG-16/19). We profile the Nvidia Jetson Nano to obtain each task's local execution time $d_i^l$, local processing power $p_i^l$, and offloading power $p_i^o$. Based on $d_i^l$, we set the period of $n_i$ to $50$, $67$, or $100$ ms, corresponding to task operating frequencies (i.e., $1/$period) of $20$, $15$, or $10$ Hz, respectively.

\textit{Benchmark Algorithms}. We compare $\mathtt{SARound}$ against four algorithms derived from state-of-the-art research on DOAP in edge computing systems \cite{gao2022deadline, hu2020dynamic, li2019cooperative, qi2024joint, fan2023joint, gao2025local}: $\mathtt{Greedy}$, $\mathtt{Iterative}$, $\mathtt{Game}$, and $\mathtt{IDAssign}$. Among these, only $\mathtt{IDAssign}$ offers a theoretical guarantee of $\frac{1}{6}$, whereas the other three benchmark algorithms do not provide any theoretical guarantees for problem $\mathbf{P}$.
\begin{itemize}
    \item \textbf{$\mathtt{Greedy}$} \cite{gao2022deadline}: For each $\ell \in \mathcal{L}$, $\mathtt{Greedy}$ defines resource efficiency as $\psi(\ell) = u(\ell)/(\frac{b_{\ell}}{B_\ell} \times \frac{c_{\ell}}{C_\ell})$, where $B_\ell$ and $C_\ell$ are the resource capacities of $m_\ell$. It then processes all $\ell \in \mathcal{L}$ in non-increasing order of $\psi(\ell)$, selecting each service whenever feasibility constraints are met.
    \item \textbf{$\mathtt{Iterative}$} \cite{hu2020dynamic, li2019cooperative, qi2024joint, fan2023joint}: $\mathtt{Iterative}$ decomposes $\mathbf{P}$ into two subproblems, job offloading and resource allocation, which are solved iteratively until convergence.
    \item \textbf{$\mathtt{Game}$} \cite{li2019cooperative}: $\mathtt{Game}$ models $\mathbf{P}$ as a non-cooperative game, in which each vehicle acts as a player, and selects the $\ell \in \mathcal{L}$ that maximizes total utility in each iteration.
    \item \textbf{$\mathtt{IDAssign}$} \cite{gao2025local}: $\mathtt{IDAssign}$ classifies service instances as either light or heavy and applies a recursive algorithm that prioritizes light instances before the heavier ones.
\end{itemize}

\begin{figure*}[t]
    \centering
    \subfigure[]{\includegraphics[width=0.32\linewidth, page=1]{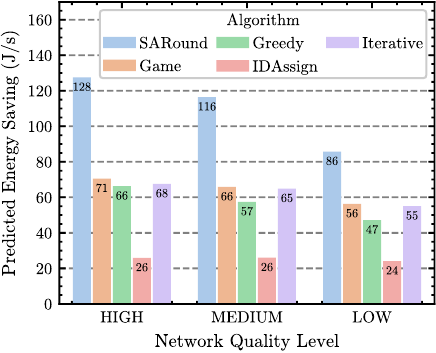} \label{fig:exp-pilot-1}} 
    \subfigure[]{\includegraphics[width=0.32\linewidth, page=2]{figures/exp-pilot-all.pdf} \label{fig:exp-pilot-2}}
    \subfigure[]{\includegraphics[width=0.32\linewidth, page=3]{figures/exp-pilot-all.pdf} \label{fig:exp-pilot-5}} 
    \caption{In \textit{SchedAll} mode, average results for: (a) \textit{predicted} energy saving; (b) \textit{measured} energy saving; (c) number of offloaded job instances per second.}
    \label{fig:exp-pilot-energy}
\end{figure*}

\begin{figure*}[t]
    \centering
    \subfigure[]{\includegraphics[width=0.32\linewidth, page=1]{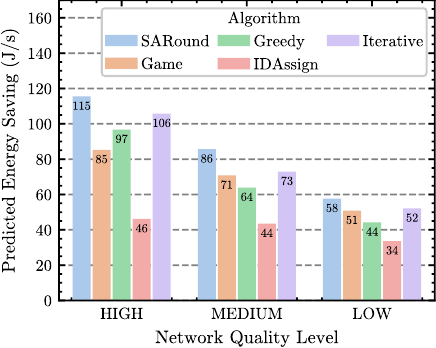} \label{fig:exp-pilot-3}} 
    \subfigure[]{\includegraphics[width=0.32\linewidth, page=2]{figures/exp-pilot-remain.pdf} \label{fig:exp-pilot-4}}
    \subfigure[]{\includegraphics[width=0.32\linewidth, page=3]{figures/exp-pilot-remain.pdf} \label{fig:exp-pilot-6}}
    \caption{In \textit{SchedRemain} mode, average results for: (a) \textit{predicted} energy saving; (b) \textit{measured} energy saving; (c) number of offloaded job instances per second.}
    \label{fig:exp-pilot-energy2}
\end{figure*}

\subsection{Result Discussion}
In this experiment, we evaluate the performance of $\mathtt{SARound}$ from two perspectives: network quality level and scheduling mode. We consider three distinct network quality levels: LOW, MEDIUM, and HIGH. The HIGH network quality level corresponds to a high channel quality with minimal fluctuations over time, while the MEDIUM and LOW levels represent progressively lower channel quality with greater fluctuations. For scheduling mode, both \textit{SchedAll} and \textit{SchedRemain} are evaluated. Algorithm performance is measured using two metrics: predicted and measured energy savings. \textit{Predicted energy saving} refers to the energy savings estimated by the scheduling algorithm, roughly representing its performance under a static network quality assumption. In contrast, \textit{measured energy saving} quantifies the energy savings measured based on the actual offloaded job instances in simulations, accounting for network fluctuations as well as the overheads of scheduling algorithm execution and service initialization. Since the simulation runs for $15$ minutes, we average the energy saving values to Joules per second (J/s) for consistency.



The average predicted and measured energy savings under various network quality levels in the \textit{SchedAll} mode are presented in Figs. \ref{fig:exp-pilot-1} and \ref{fig:exp-pilot-2}, respectively. As observed, both energy values across all algorithms decrease as channel quality deteriorates. This is primarily because a lower channel quality leads to a smaller MCS to ensure reliable data transmission. The smaller MCS lowers the data offloading rate, thereby limiting the number of requests that can be accommodated within the same bandwidth resources. Notably, $\mathtt{SARound}$ consistently outperforms all benchmark algorithms in both predicted and measured energy saving across all network quality levels. When normalizing energy values based on the average results of $\mathtt{SARound}$ across all network quality levels, the predicted energy saving (Fig. \ref{fig:exp-pilot-1}) of $\mathtt{SARound}$ exceeds that of $\mathtt{Game}$ by $41.54\%$, $\mathtt{Greedy}$ by $48.16\%$, $\mathtt{IDAssign}$ by $76.94\%$, and $\mathtt{Iterative}$ by $43.10\%$. Similarly, the measured energy saving (Fig. \ref{fig:exp-pilot-2}) of $\mathtt{SARound}$ surpasses $\mathtt{Game}$ by $22.98\%$, $\mathtt{Greedy}$ by $42.95\%$, $\mathtt{IDAssign}$ by $75.09\%$, and $\mathtt{Iterative}$ by $20.26\%$. These results demonstrate the superior resource utilization efficiency of $\mathtt{SARound}$ in the \textit{SchedAll} mode. 

The average predicted and measured energy saving under different network quality levels in the \textit{SchedRemain} mode are presented in Figs. \ref{fig:exp-pilot-3} and \ref{fig:exp-pilot-4}, respectively. $\mathtt{SARound}$ maintains superior performance in both predicted and measured energy savings compared to all benchmark algorithms. 
On average, the predicted energy saving of $\mathtt{SARound}$ exceeds that of $\mathtt{Game}$ by $20.04\%$, $\mathtt{Greedy}$ by $20.78\%$, $\mathtt{IDAssign}$ by $52.34\%$, and $\mathtt{Iterative}$ by $10.79\%$. Similarly, the measured energy saving of $\mathtt{SARound}$ outperforms $\mathtt{Game}$ by $22.84\%$, $\mathtt{Greedy}$ by $13.47\%$, $\mathtt{IDAssign}$ by $50.54\%$, and $\mathtt{Iterative}$ by $8.25\%$. It is observed that the measured energy saving exceeds the predicted energy saving when the network quality level is HIGH for $\mathtt{SARound}$, $\mathtt{Greedy}$, and $\mathtt{IDAssign}$. This occurs because the channel quality used for scheduling is not always optimal; an improved channel quality after scheduling can enhance energy savings for the same RB allocation.
Combining the results from Figs. \ref{fig:exp-pilot-1}, \ref{fig:exp-pilot-2}, \ref{fig:exp-pilot-3}, and \ref{fig:exp-pilot-4}, $\mathtt{SARound}$ consistently achieves higher energy savings than baselines across all network quality levels and scheduling modes, demonstrating its effectiveness in resource allocation. 

\begin{figure*}[t]
    \centering
    \subfigure[]{\includegraphics[width=0.43\linewidth, page=1]{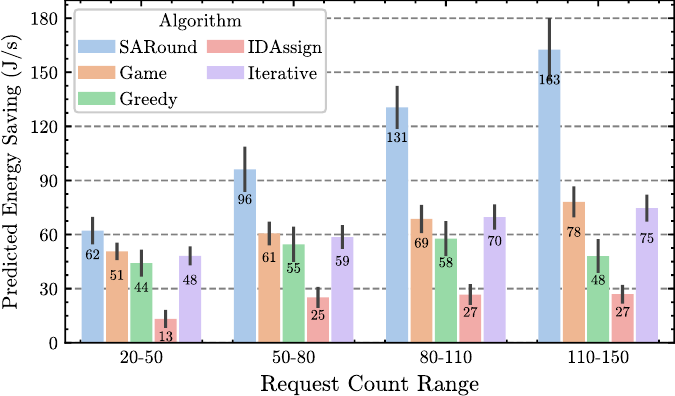} \label{fig:exp-count-1}} 
    \hspace{1cm}
    \subfigure[]{\includegraphics[width=0.42\linewidth, page=2]{figures/exp-count.pdf} \label{fig:exp-count-2}}
    \caption{For different request counts, the results for (a) predicted energy saving and (b) algorithm execution time.}
    \label{fig:exp-count}
\end{figure*}

\textit{Impact of Scheduling Mode on $\mathtt{SARound}$.} Across all network quality levels, $\mathtt{SARound}$ achieves a higher predicted energy saving in the \textit{SchedAll} mode compared to the \textit{SchedRemain} mode (Figs. \ref{fig:exp-pilot-1} and \ref{fig:exp-pilot-3}). This is because, in the \textit{SchedAll} mode, all requests are scheduled using the full RSU resource, 
enabling $\mathtt{SARound}$ to implement finer-grained resource allocation in the \textit{SchedAll} mode.
However, despite the lower predicted energy saving in the \textit{SchedRemain} mode, its measured energy saving is higher (Figs. \ref{fig:exp-pilot-2} and \ref{fig:exp-pilot-4}). 
This is because, in the \textit{SchedRemain} mode, running services are not terminated as long as task deadlines can still be met. Over time, this leads to an accumulation of services with longer slack times (i.e., the time between job completion and its deadline) that are more tolerant to network fluctuations, resulting in fewer job instance rejections. Furthermore, these services avoid delays caused by resource rescheduling and service re-initializing, allowing more job instances to be offloaded.


The average number of offloaded job instances per second in the \textit{SchedAll} and \textit{SchedRemain} modes is illustrated in Figs. \ref{fig:exp-pilot-5} and \ref{fig:exp-pilot-6}, respectively. In the \textit{SchedAll} mode, $\mathtt{SARound}$ consistently offloads more jobs than all benchmark algorithms across all network quality levels, and the performance advantage increases as network quality improves. While this advantage is less pronounced in the \textit{SchedRemain} mode, $\mathtt{SARound}$ still achieves a higher average number of offloaded jobs compared to the benchmark algorithms. Besides, for any network quality level, more job instances are offloaded in the \textit{SchedRemain} mode compared to the \textit{SchedAll} mode for all algorithms. This is because, in the \textit{SchedRemain} mode, running services are not terminated at the end of a scheduling cycle. This reduces unnecessary overheads for scheduling algorithm execution and service initialization in each new scheduling cycle, ensuring a higher overall offloaded job count. 

Lastly, we analyze how the predicted energy saving and execution time of different scheduling algorithms vary with the number of pending requests. To ensure consistent RSU resource availability, we select the \textit{SchedAll} mode and maintain the network quality level at MEDIUM. The average predicted energy saving is shown in Fig. \ref{fig:exp-count-1}, while the average algorithm execution time is presented in Fig. \ref{fig:exp-count-2}. 
As illustrated in Fig.~\ref{fig:exp-count-1}, $\mathtt{SARound}$ outperforms all benchmark algorithms across various request count ranges, with its performance advantage increasing significantly as the request count grows. This improvement is primarily due to the diminishing impact of LP rounding loss in $\mathtt{SARound}$ as the number of requests increases, leading to enhanced performance. 
Fig.~\ref{fig:exp-count-2} shows that the average execution time of $\mathtt{SARound}$ is approximately $2.1\times$ that of $\mathtt{Game}$, $1.3\times$ that of $\mathtt{Greedy}$, $0.9\times$ that of $\mathtt{IDAssign}$, and $5.5\times$ that of $\mathtt{Iterative}$. Notably, the execution time of $\mathtt{SARound}$ increases linearly with the number of requests, indicating strong scalability. Moreover, within our framework, a longer algorithm runtime only postpones task offloading without affecting the execution window of individual job instances. Therefore, it is practical to deploy $\mathtt{SARound}$ in real-world systems.


\textit{Conclusion.} In this experiment, we define the task utility function as the energy saving for vehicles. Experimental results demonstrate that $\mathtt{SARound}$ consistently outperforms benchmark heuristics in both predicted energy saving and measured energy saving across various network quality levels and scheduling modes. Additionally, $\mathtt{SARound}$ admits more job instances during the simulation, highlighting its superior resource allocation efficiency and fairness. 
Lastly, $\mathtt{SARound}$ achieves improved performance with increasing request counts and exhibits a linearly increasing execution time with the request count in practice, underscoring its runtime efficiency.

\section{Conclusion}\label{sec:conclusion}
This paper addressed a DOAP in VEC, problem $\mathbf{P}$, which jointly considered task offloading and resource allocation under both wireless bandwidth and computational resource constraints, aiming to maximize the total utility for vehicles. To address $\mathbf{P}$, we proposed an approximation algorithm, $\mathtt{SARound}$, which improved the best-known approximation ratio from $\frac{1}{6}$ to $\frac{1}{4}$ while maintaining runtime efficiency. To address the challenges of short task deadlines and changing wireless network conditions in VEC, we introduced an online service subscription and offloading control framework that enables dynamic task life-cycle management. Additionally, we developed a novel VEC simulator, VecSim, integrating the full framework and used it to evaluate algorithm performance. Results showed that $\mathtt{SARound}$ significantly outperformed existing baselines
across various network conditions and scheduling modes.

In this work, we focus on global resource optimization in VEC systems under a centralized scheduling architecture. While this design facilitates coordinated decision-making, it also introduces a single point of failure. As future work, we plan to explore decentralized architectures for VEC to enhance system robustness against both scheduler failures and communication disruptions between RSUs and the scheduler.

\newpage

\bibliographystyle{IEEEtran}
\bibliography{IEEEabrv,reference}

\end{document}